\documentclass[preprint,12pt]{article}

\usepackage[english]{babel}	
\usepackage[utf8x]{inputenc}		
\usepackage{enumerate}
\usepackage{textcomp}                
\usepackage{typearea}
\usepackage{pdfpages}
\usepackage[toc]{appendix}
\usepackage[disable]{todonotes}
\usepackage{subcaption}
\usepackage{mathtools}

\usepackage{amsfonts, amsmath, amssymb, amsthm, dsfont}
\usepackage{amsthm}			
\usepackage{thumbpdf}	
\usepackage{natbib}
\usepackage{booktabs}

\usepackage{pgf, tikz}
\usepackage{tikz-cd}		
\usepackage{bbm}

\theoremstyle{definition}
\newtheorem{definition}{Definition}[section]

\newtheorem*{assumption*}{Assumption}

\newtheorem*{condition*}{Condition}

\theoremstyle{plain}
\newtheorem{theorem}[definition]{Theorem}
\newtheorem{proposition}[definition]{Proposition}
\newtheorem{lemma}[definition]{Lemma}

\theoremstyle{remark}

\usepackage{graphicx}
\usepackage{hyperref}
\allowdisplaybreaks

\newcommand{\N}{\mathbb{N}}

\newcommand{\R}{\mathbb{R}}

\newcommand{\E}{\mathbb{E}}

\newcommand{\Var}{\operatorname{Var}}
\newcommand{\Cov}{\operatorname{Cov}}

\newcommand{\deq}{\overset{d}{=}}

\newcommand{\diag}{\mathrm{diag}}

\newcommand{\tr}{\text{tr}}

\newcommand{\beps}{\boldsymbol{\epsilon}}

\newif\ifhideproofs
\ifhideproofs
  \usepackage{environ}
  \NewEnviron{hide}{}

\fi

\title{Projection inference for high-dimensional covariance matrices with structured shrinkage targets}
\author{Fabian Mies\and Ansgar Steland}
\date{RWTH Aachen University\\[2ex] \today}

\begin{document}
\maketitle

\begin{abstract}
    Analyzing large samples of high-dimensional data under dependence is a challenging statistical problem as long time series may have change points, most importantly in the mean and the marginal covariances, for which one needs valid tests. Inference for large covariance matrices is especially difficult due to noise accumulation, resulting in singular estimates and poor power of related tests. The singularity of the sample covariance matrix in high dimensions can be overcome by considering a linear combination with a regular, more structured target matrix.
	This approach is known as shrinkage, and the target matrix is typically of diagonal form. 
	In this paper, we consider covariance shrinkage towards structured nonparametric estimators of the bandable or Toeplitz type, 
 respectively, aiming at improved estimation accuracy and statistical power of tests even under nonstationarity.
	We derive feasible Gaussian approximation results for bilinear projections of the shrinkage estimators which are valid under nonstationarity and dependence.
	These approximations especially enable us to formulate a statistical test for structural breaks in the marginal covariance structure of high-dimensional time series without restrictions on the dimension, and which is robust against nonstationarity of nuisance parameters.
	We show via simulations that shrinkage helps to increase the power of the proposed tests.
	Moreover, we suggest a data-driven choice of the shrinkage weights, and assess its performance by means of a Monte Carlo study.
	The results indicate that the proposed shrinkage estimator is superior for non-Toeplitz covariance structures close to fractional Gaussian noise. 
 
	\textbf{Keywords:} structural break; Gaussian approximation; nonstationary time series; bilinear form
\end{abstract}

\section{Introduction}

Many modern applications encounter large data sets of very high dimension $d$, which requires new approaches for modeling as well as statistical inference, since classical methods usually fail. A particularly prominent example is the covariance matrix $\Sigma$ of a random vector of $d$ variables, which has $d^2$ unknown parameters.
It is well-known that the sample covariance matrix $\hat{\Sigma}_n$ of sample size $n$ is singular if $n<d$, and consistency of $\hat{\Sigma}_n$ for growing dimension $d$ requires to control the $O(d^2) $ covariances and otherwise generally fails. Specifically, for $n$ i.i.d.\ samples of $d$-dimensional random vectors with finite fourth moments, consistency in terms of the expected squared Frobenius norm of $\hat{\Sigma}_n- \Sigma_n $ requires that $ \diag( \Sigma_n ) $ and the variance of the average of the squared variables to be of the order $ o(n/d) $, \citep{Ledoit2004}. In terms of eigenvalues and eigenvectors, classic inconsistency results for the largest eigenvalue and the top eigenvector under the regime $ d/n \to c \in (0,1) $ are due to \citep{Johnstone2001,Johnstone2009}, even when $ \Sigma = I $. Since, however, it is generally believed that many high-dimensional data sets are governed by, say $r \ll d$, strong signals, spiked covariance models with $r$ leading eigenvalues strictly larger than $1$, whereas the remaining bulk of the spectrum consists of unit eigenvalues, are quite extensively studied. Here the leading eigenvalues of $ \hat{\Sigma}_n $ are biased, thus requiring de-biasing procedures which, however, depend on the loss function, and the eigenvectors are inconsistent as well, see \citep{DonohoGavinJohnstone2018} and the references therein. 
Further potential remedies are to impose additional structural assumptions on $\Sigma$, e.g.\ sparsity \citep{bickel2008a}, Toeplitz shape \citep{cai2013}, or assuming the entries further from the diagonal to decay rapidly, also known as a bandable covariance matrix \citep{bickel2008}.
In these frameworks, it is possible to derive consistent estimators of $\Sigma$ even if $d\gg n$, see \cite{Cai2016} for a survey of suitable estimators and matching minimax bounds.

The structured covariance matrix estimators admit optimal rates of convergence within the specified class of matrices.
However, if the true $\Sigma$ does not satisfy the structural assumptions, the estimators might perform poorly.
For instance, if we use an estimator of Toeplitz form, but $\Sigma$ is not a Toeplitz matrix, then this estimator will be inconsistent.
In order to strike a compromise between structured and model-free estimation, one may consider a shrinkage estimator of the form 
\begin{align*}
    \hat{\Sigma}^w_n = (1-w)\,\hat{\Sigma}_n\; + \; w\, \tilde{\Sigma}_n, \qquad w\in [0,1],
\end{align*}
where $\hat{\Sigma}_n$ is the usual sample covariance matrix, and the shrinkage target $\tilde{\Sigma}_n$ is a structured estimator. 
The weight $w$ may be interpreted as quantifying the confidence in the structural assumption underlying the estimator $\tilde{\Sigma}$. 
A particularly useful feature of shrinkage estimators is that $\hat{\Sigma}_n^w$ will be regular, as $\tilde{\Sigma}_n$ is usually chosen as a regular matrix. 
This makes shrinkage estimation especially interesting for situations where the inverse of the covariance matrix is required, e.g.\ in portfolio optimization.
\cite{Ledoit2004} suggest to use a multiple of the identity matrix, and \cite{Steland2018a} studies more general diagonal matrices.
Linear shrinking towards a structured target with respect to the scaled Frobenius loss has been studied in \citep{LedoitWolf2003}, but the results therein on consistent estimators of the optimal shrinking weights consider classical fixed-$d$ asymptotics only. In the present paper, we go beyond a convex combination of the sample covariance matrix with a single shrinkage target. Instead, we study convex combinations with a structured covariance estimator with known minimax optimality properties for bandable $ \Sigma $, and a structured covariance estimator which is minimax-optimal if $ \Sigma $ is of Toeplitz type. 

Theoretical analysis of these structured covariance estimators typically focuses on their consistency and rate of convergence. Here, we contribute by showing that the estimators proposed when $ \Sigma $ is bandable and Toeplitz, respectively, are still consistent under a general nonstationary  nonlinear time series model, whereas the known results assume iid $d$-dimensional random vectors. 
Moreover, to perform statistical inference, distributional approximations are required. 
Here, we consider inference about the high-dimensional covariance matrix in terms of projections $v^T\Sigma v$ for some vector $v\in\R^d$.
The vector $v$ may either be determined by the application, or chosen at random, independently from the sample.
For example, in risk management, $v$ may represent weights of a portfolio of risky assets, and $v^T\Sigma v$ corresponds to the variance of this portfolio. 
As another example, the hypothesis $\Sigma = \Sigma_0$ may be studied via random projections: 
If $v$ is drawn from some continuous distribution, then $\Sigma\neq \Sigma_0$ implies $v^T\Sigma v \neq v^T\Sigma_0 v$ almost surely.
It is thus reasonable to employ a test statistic of the form $v^T(S-\Sigma_0) v$ for some estimator $S$, either structured, unstructured, or shrunken. Further applications (lasso prediction, sparse principal components) are discussed below. 
The idea to study high-dimensional covariance matrices via bilinear forms has been suggested in \cite{Steland2017, Steland2018}, and applied to changepoint testing \citep{Steland2019} and to a K-sample problem \cite{Mause2020}.
In the latter references, as well as in the present paper, the important assumption is that the projection vectors $v$ have a bounded $l_1$-norm. This condition is satisfied in many applications and enables the methodology to work also for very high dimensions $d$.

In this paper, we derive Gaussian approximation results for shrinkage estimators of the covariance matrix with structured shrinkage targets $\tilde{\Sigma}_n$. 
Our results are valid for nonstationary, nonlinear, high-dimensional time series $X_t$, and we demonstrate how to perform inference based on a suitable bootstrap scheme.
Notably, we only impose assumptions on the marginal time series, while the dependency between the components may be arbitrary, and we need no restrictions on the dimension $d$. 
This is a beneficial consequence of the projection technique, as described in Section \ref{sec:projection}. 
As we allow the dimension to grow with sample size, the considered models are arrays of time series.
Thus, classical central limit theory is not applicable because a limiting distribution in general does not exist.
Instead, we formulate our results in terms of so-called strong approximations, or sequential couplings, in the sense of \cite{komlos1975}.
Compared to the existing work of \cite{Steland2017, Steland2018}, we consider a much broader class of nonlinear time series models, which also allows us to study non-diagonal, structured shrinkage targets.
Underlying our mathematical results are recent Gaussian couplings for nonstationary time series established in \cite{mies2022} for the regime $d\ll n^\frac{1}{3}$.
We extend these results to projections $v^T X_t$ of very high-dimensional time series $X_t$.
In particular, it is shown that for $l_1$-bounded projection vectors $v$, the Gaussian approximation is valid irrespective of the ambient dimension $d$, and of the dependence structure between the $d$ components.
It turns out that this general projection framework is also applicable to bilinear forms of sample covariance matrices, as demonstrated in Section \ref{sec:coupling-cov}.
The presented distributional approximation also holds sequentially, and we demonstrate how to apply this to a test for structural breaks in the covariance matrix. 
An attractive feature of the proposed test is that it is robust against general nonstationarity under the null hypothesis, and only detects changes in the target parameter, i.e.\ the marginal covariance structure.
Furthermore, we discuss the optimal choice of shrinkage weights and provide a data-driven criterion.
Simulation results demonstrate that the asymptotic theory is also applicable in finite samples, and that shrinkage may improve the performance both for estimation and for change testing.

The rest of this paper is structured as follows. 
In Section \ref{sec:projection}, we present a general sequential Gaussian approximation result for projections of high-dimensional nonstationary time series, which is applied to bilinear forms of covariance matrices in Section \ref{sec:coupling-cov}. 
The data-driven choice of shrinkage weights is discussed in Section \ref{sec:optimal}.
Section \ref{sec:changepoint} describes the application of our distributional approximations to tests for changes in the covariance structure of a high-dimensional time series, and presents a feasible bootstrap scheme.
Simulation results are presented in Section \ref{sec:simulation}.

\section{Gaussian approximations for projections}\label{sec:projection}

We consider a $d$-variate time series $X_t$ which may be nonlinear and nonstationary, given by
\begin{align}
	X_t = G_t(\epsilon_t,\epsilon_{t-1},\ldots), \label{eqn:ts-model}
\end{align}
for measurable mappings $G_{t}:\R^\infty\to\R^{d}$, $t=1,\ldots, n$, where we endow $\R^\infty$ with the $\sigma$-algebra generated by all finite projections.
The $\epsilon_i$ are iid $U[0,1]$ random variables, which may be regarded as random seeds for the time series.
The model formulation \eqref{eqn:ts-model} allows for a convenient formulation of ergodicity conditions via the physical dependence measure introduced by \cite{Wu2005}.
To formulate our assumptions on the kernels $G_t$, introduce a second sequence $\tilde{\epsilon}_i$ of iid $U[0,1]$ random variables, independent of the sequence $\epsilon_i$.
For any $t\in\N$, denote 
\begin{align*}
	\beps_t &= (\epsilon_t, \epsilon_{t-1},\ldots) \in\R^\infty, \\
	\tilde{\beps}_{t,j} &= (\epsilon_t,\ldots, \epsilon_{j+1}, \tilde{\epsilon}_j,  \epsilon_{j-1},\ldots) \in\R^\infty, 
\end{align*}
such that $X_t = G_t(\beps_t) = (G_{t,l}(\beps_t))_{l=1}^d$.
We assume that the impact of past random seeds decays to zero polynomially, such that for some $\beta>1$
\begin{align}
	\begin{split}
		\left(\E |G_{t,l}(\beps_t) - G_{t,l}(\tilde{\beps}_{t,t-j})|^q\right)^\frac{1}{q} &\leq \tilde{\Theta} j^{-\beta},\qquad  j\geq 0,\quad l=1,\ldots,d,\\
		\left(\E |G_{t,l}(\beps_0)|^q\right)^\frac{1}{q} &\leq \tilde{\Theta},\qquad l=1,\ldots,d. 
	\end{split}\label{eqn:ass-ergodic-proj} \tag{P.1}
\end{align}

The time series $X_t$ may be nonstationary, and the nonstationarity is explicit by making the kernel $G_t$ depend on $t$. 
Nevertheless, to obtain stronger asymptotic results, we require some regularity in time. 
Here, we do not choose classical smoothness conditions, but rather formulate the regularity in terms of the total variation norm of the mapping $t\mapsto G_t$, measured in $L_2(P)$.
In particular, we suppose that for some $\tilde{\Gamma}>0$,
\begin{align}
	\sum_{t=2}^{n}\left( \E |G_{t,l}(\beps_0) - G_{t-1,l}(\beps_0)|^2 \right)^\frac{1}{2} 
	&\leq \tilde{\Gamma}\cdot\tilde{\Theta},\quad l=1,\ldots,d. \label{eqn:ass-BV-proj} \tag{P.2}
\end{align}

A few results require to assume
\[
 \text{$ X_t $ has finite and stationary eighth moment and $ \beta > 2 $.} \tag{P.3}
\]

We highlight that \eqref{eqn:ass-ergodic-proj} and \eqref{eqn:ass-BV-proj} only impose conditions on the individual components $l=1,\ldots,d$, and are thus rather easy to verify even for high-dimensional models. To emphasize, we impose no assumption on the dependency among the $d$ components.

Now consider the projection of the $d$-variate time series $X_t$ to a $m$-dimensional time series $Z_t$, with $m\ll d$.
That is, for a matrix $V\in\R^{m\times d}$, we consider the time series $Z_t = V X_t$.
We measure the size of the matrix $V$ by the operator norm, $\|V\|_\infty$, with respect to the maximum vector norm, i.e.\
\begin{align*}
	\|V\|_\infty = \sup_{w\in\R^d} \frac{\|Vw\|_\infty}{\|w\|_\infty} = \max_{l=1,\ldots,m} \|V_{l,\cdot}\|_1,
\end{align*}
where $V_{l,\cdot}$ denotes the $l$-th row of the matrix $V$. This choice of norm uniformly controls the $ \ell_1 $-norms of the projection vectors and is motivated by highdimensional statistical problems. Indeed, as well known, if $d$ is large compared to $n$, classical statistical methods usually fail, because even strong signals are overlayed by too many random noise sources (noise accumulation) and spurious correlations occur. To overcome the curse of high dimensionality, sparse methods typically use $ \| \cdot \|_r$-norms with $ 1 \le r < 2$, often leading to $s$-sparse solutions, i.e.\ with only $s$ active (non-zero) coordinates.  For example, sparse PCA constructs $s$-sparse directions $ u $ on which the data vectors are projected by calculating $ u^T X_t$, and the variance of such a coefficients $ u^T X_t $, given by the quadratic form $ u^T \Sigma_n u $ and coinciding with the associated eigenvalue when $ u $ is an eigenvector, provides information about the importance of that direction and is thus of interest. Similarly, when predicting a future response $\tilde{Y}_t$ by a lasso regression on $ X_t $, one calculates $ \hat{\beta}_n^T \tilde{X}_t $ for future regressors $ \tilde{X}_t $, where the lasso estimate $ \hat{\beta}_n $ is $s-$sparse and thus selects $s \ll d$ variables, so that it behaves like a low-dimensional statistic circumventing the issues arising for large dimension. The variance $ \hat{\beta}_n^T \Sigma_n \hat{\beta}_n^T $ of the prediction $ \hat{\beta}_n^T \tilde{X}_t$ is a natural measure of the prediction accuracy. When considering slightly more general vectors $ v $ with bounded $ \| v \|_1 $-norm (uniformly in $d$ and $n$), then we have the bound $ | v^T x_t | \le \| v \|_1 \max_{1 \le j \le d} | X_{tj} |$ at our disposal. Indeed, such projections have bounded moments under weak assumptions guaranteed by (P.1) as summarized in the following lemma, whose assertions no longer hold if $ \| v \|_2 $ is uniformly bounded but $ \| v \|_1 \to \infty $.

\begin{lemma}\label{lem:jensen} Let $ q \ge 1 $ and suppose that $ \max_{1 \le j \le d} \E |X_{tj}|^q \le \tilde{\Theta} $ for some constant $C$ and all $d, n$, cf. (P.1). Then for any $ v \in \R^s $ with norm $ \| v \|_1 $ uniformly bounded by some $M$ 
\[
  \left(\E | v^T X_t |^q\right)^{\frac{1}{q}} \le \| v \|_1 \max_{1 \le j \le d} ( \E |X_{tj}|^q )^{\frac{1}{q}} \le \tilde{\Theta} M
\]
\end{lemma}

The formulation of our first main result requires the two rates
\begin{align*}
	\chi(q,\beta) = \begin{cases}
		\frac{q-2}{6q-4}, & \beta \geq \frac{3}{2}, \\
		\frac{(\beta-1)(q-2)}{q(4\beta-3)-2}, &\beta\in(1,\frac{3}{2}).
	\end{cases},
	\quad
	\xi(q,\beta) = \begin{cases}
		\frac{q-2}{6q-4}, & \beta \geq 3, \\
		\frac{(\beta-2)(q-2)}{(4\beta-6)q-4},& \frac{3+\frac{2}{q}}{1+\frac{2}{q}} < \beta < 3, \\
		\frac{1}{2}-\frac{1}{\beta}, & 2< \beta \leq \frac{3+\frac{2}{q}}{1+\frac{2}{q}}.\\
	\end{cases}
\end{align*}

\begin{theorem}\label{thm:Gauss-ts-proj}
	Let $X_t=G_t(\beps_t)$ with $\E(X_t)=0$ be such that \eqref{eqn:ass-ergodic-proj} holds for some $q>2$, $\beta>1$, and let $Z_t = V X_t$ for some $V\in\R^{m\times d}$ such that $m\leq cn$ for some $c>0$.
	Then, on a potentially different probability space, there exist random vectors $(X_t')_{t=1}^n\deq (X_t)_{t=1}^n$ and independent, mean zero, Gaussian random vectors $Y_t'$ such that
	\begin{align}
		\left(\E \max_{k\leq n} \left\|\frac{1}{\sqrt{n}} \sum_{t=1}^k  (V X_t' - Y_t') \right\|_2^2 \right)^\frac{1}{2}
		&\leq C \tilde{\Theta} \|V\|_\infty \sqrt{m \,\log(n)}  \left( \frac{m}{n} \right)^{\chi(q,\beta)}. \label{eqn:Gauss-ts-1-proj}
	\end{align} 
	for some universal constant $C=C(q,\beta,c)$.
	
	If $\beta>2$, the local long-run variance $\Xi_t(v) = \sum_{h=-\infty}^\infty V\Cov(G_t(\beps_0),G_t(\beps_h))V^T$ is well defined.
	If \eqref{eqn:ass-BV-proj} is satisfied aswell, then there exist random vectors $(X_t')_{t=1}^n\deq (X_t)_{t=1}^n$ and independent, mean zero, Gaussian random variables $Y_t'\sim\mathcal{N}(0,\Xi_t(v))$ such that
	\begin{align}
		\begin{split}
			\left(\E \max_{k\leq n} \left\| \frac{1}{\sqrt{n}}\sum_{t=1}^k  (V X_t' - Y_t') \right\|_2^2 \right)^\frac{1}{2} 
			&\leq C \tilde{\Theta} \|V\|_\infty \sqrt{m \,\log(n)}  \left( \frac{m}{n} \right)^{\xi(q,\beta)}.
		\end{split}
		\label{eqn:Gauss-ts-2-proj}
	\end{align}
\end{theorem}

The rate of the approximation \eqref{eqn:Gauss-ts-1-proj} is faster than \eqref{eqn:Gauss-ts-2-proj}. 
The difference is that in the first case, the covariance structure of the approximating Gaussian random vectors $Y_t'$ is not explicit.
In the second case, the distribution of the approximating random vectors is explicit, at the price of assuming some temporal regularity of the stochastic process $X_t$.
In Section \ref{sec:changepoint}, we will describe a bootstrap scheme to perform inference based on Theorem \ref{thm:Gauss-ts-proj}.

Neglecting the high-dimensional context for a moment, the projection $Z_t=V X_t$ may be analyzed as a multivariate, nonstationary time series of fixed dimension $m$. 
In this setting, \cite{Wu2011} and \cite{Karmakar2020} study sequential Gaussian approximations, the latter obtaining optimal rates in $n$ which are faster than the rates of Theorem \ref{thm:Gauss-ts-proj}. 
However, they require a lower bound on the covariance matrices of $Z_t$, while our result does not need this assumption. 
This is particularly relevant because we do not impose any conditions on the dependency among the components of $X_t$. 
As an extreme example, if all $d$ components are independent, then a $l_1$-bounded projection will lead to some concentration due to averaging, such that a lower bound on the variance of the projection can in general not be guaranteed.
This prevents an application of the result of \cite{Karmakar2020} in the present situation. 
Moreover, the covariance of the approximating random vectors in \cite{Karmakar2020} is not explicit.
Instead, we employ a result on Gaussian couplings presented in \cite{mies2022}, which is in turn enabled by a recent result of \cite{eldan2020}.

\section{Asymptotics of high-dimensional shrinkage covariance matrix estimators and their projections}\label{sec:coupling-cov}

For a $d$-variate centered time series $X_t$ as above, denote the sample covariance matrix by $\hat{\Sigma}_{n,n} = \frac{1}{n}\sum_{t=1}^n X_t X_t^T$, the partial sums by $\hat{\Sigma}_{k,n} = \frac{1}{n}\sum_{t=1}^k X_t X_t^T$, $k=1,\ldots, n$, and its mean by
\begin{align*}
    \Sigma_n \;=\; \E \hat{\Sigma}_{n,n}  \;=\; \frac{1}{n} \sum_{t=1}^n \Cov(X_t).
\end{align*}
Hence, we account for nonstationarity of the time series $X_t$ by averaging the covariance matrices at all times $t=1,\ldots, n$.

The central observation enabling our subsequent analysis is that the matrix-valued time series $X_t X_t^T$ fits within the framework of Section \ref{sec:projection}.

\begin{proposition}\label{prop:matrix}
	If the time series $X_t = G_t(\beps_t)$ satisfies \eqref{eqn:ass-ergodic-proj} resp.\ \eqref{eqn:ass-BV-proj} with power $2q$ and factor $\tilde{\Theta}$, then the $d^2$-variate time series $\check{X}_t = X_t X_t^T = \check{G}_t(\beps_t)$ satisfies \eqref{eqn:ass-ergodic-proj} resp.\ \eqref{eqn:ass-BV-proj} with power $q$ and factor $\tilde{\Theta}^2$.
\end{proposition}

It has been suggested by \cite{Steland2017} to perform inference on the covariance matrix via the quadratic form $v^T \hat{\Sigma}_{n,n} v$ for some vector $v\in\R^d$, i.e., the variance of the projection $ v^T X_t$. As already discussed above, such projections are ubiquitous in statistical problems, for example arising in optimal portfolio selection, as linear predictors in regression models, or when reducing dimensionality by principal component analysis. The associated empirical version, $v^T \hat{\Sigma}_{n,n} v = \frac{1}{n}\sum_{t=1}^n |v^T X_t|^2$, may be interpreted as the estimated variance of the univariate projection $\tilde{X}_t=v^TX_t$. The results of \cite{Steland2017} and \cite{Steland2018} on Gaussian approximations of $ v^T \hat{\Sigma}_n v $, related change-point procedures, and Shrinkage estimators, impose the assumption that the process $X_t$ is a linear time series with a single innovation process, which may be interpreted as a special type of single-factor model, although approximate vector autoregressions and spiked covariance models are included under regularity conditions, see \cite{Steland2019}. Extensions to multi-factor models with infinitely many factors and general multivariate linear processes with finite-dimensional innovations are studied in \cite{Bours2021}. Although such processes, especially factor models, provide good description for many application scenarios and are widely used, they nevertheless restrict the dependence structure of the $d$ component processes and rule out phenomena such as conditional heteroscedasticity and nonlinearity. Here, by relying on the results of the previous section, we allow for nonstationary nonlinear time series and impose much weaker assumptions on the dependence of the components.

To estimate $ \Sigma_n $ and related quadractic forms $ v^T \Sigma_n v $, we shall shrink the nonparametric estimator $ \hat{\Sigma}_{n,n} $ towards structured targets, in order to overcome
various issues of the sample covariance matrix arising if $d\gg n$.  A covariance matrix $\Sigma$ is called bandable if $\Sigma_{l,l'} \leq g(|l-l'|)$ for some function $g(x)\to 0$ as $x\to\infty$.
In this situation, optimal rates of convergence can be achieved via thresholding, see \cite[Thm.\ 7]{Cai2016} and references therein.
In particular, if $g(x) \propto |x|^{-\alpha-1}$ for some $\alpha>0$, then a rate-optimal estimator is given by the tapering estimator
\begin{align*}
	\hat{\Sigma}_{k,n}^{\dagger} &= ((\hat{\Sigma}_{k,n})_{l,l'} \omega(|l-l'|))_{l,l'=1}^d, \\
	\omega(x) &= \begin{cases}
		1, & x\leq \frac{\tau}{2}, \\
		2-\frac{2x}{\tau}, & \frac{\tau}{2}< x \leq \tau, \\
		0, & x>\tau,
	\end{cases}
\end{align*}
The optimal rate is achieved by the threshold $\tau=\min(n^\frac{1}{2\alpha+c},d)$ with $ c = 2 $ under the Frobenius norm $\|\cdot\|_F$ and with $c=1$ under the spectral norm $\|\cdot\|$, see \cite{Cai2010}. 
In particular, the optimal rates are
\begin{align*}
    d^{-1}\|\widehat{\Sigma}_n - \Sigma\|_F^2 
    = \mathcal{O}\left( n^{-\frac{2\alpha+1}{2\alpha+2}} + \frac{d}{n} \right), \qquad
    \|\widehat{\Sigma}_n - \Sigma\|^2 
    = \mathcal{O}\left( n^{-\frac{2\alpha}{2\alpha+1}} + \frac{\log d}{n} \right).
\end{align*}

Another approach is to impose a Toeplitz structure, i.e.\ $\Sigma_{l,l'} = \sigma_{|l-l'|}$, for some sequence $ \sigma_0, \sigma_1, \ldots $. The corresponding tapered Toeplitz estimator of $\Sigma$ is 
\begin{align*}
	\hat{\Sigma}_{k,n}^\diamond &= \left( \hat{\sigma}_{|l-l'|} \omega(|l-l'|) \right)_{l,l'=1}^d, \\
	\hat{\sigma}_m &= \frac{1}{d-m} \sum_{l-l'=m} (\hat{\Sigma}_{k,n})_{l,l'}, \quad 0 \le m < d,
\end{align*}
with tapering weights $\omega(x)$ as above.
If $|\sigma_{m}|\leq C m^{-\alpha-1}$, then the optimal rate under the spectral norm is achieved by the threshold choice $\tau = (nd / \log(nd))^\frac{1}{2\alpha+1}$, see \cite[Thm.\ 9]{Cai2016} and \cite{cai2013}.

The following theorem shows that the full-sample estimators $ \hat{\Sigma}_n^\dagger = \hat{\Sigma}_{n,n}^\dagger $ and $ \hat{\Sigma}_n^\diamond = \hat{\Sigma}_{n,n}^\diamond $ are consistent in the (rescaled) Frobenius norm, under mild regularity conditions. To the best of our knowledge, these estimators have not yet been studied under such a general nonstationary nonlinear time series model as given by (P.1) and (P.2). Moreover, $\hat{\Sigma}_n^\dagger$ is rate-optimal for bandable covariance matrices. Note that the optimal rates and thresholds have been obtained under the assumption that $ X_t $ are iid random vectors with mean zero and a bandable or Toeplitz covariance matrix $ \Sigma $ satisfying further regularity conditions. 

Denote by $ \Sigma_n^\dagger $ and $ \Sigma_n^\diamond $ the theoretical oracle estimators associated to the estimators $ \hat{\Sigma}_{n,n}^{\dagger} $ and $ \hat{\Sigma}_{n,n}^\diamond $, respectively, obtained by replacing in their definitions the estimates $ (\hat{\Sigma}_{n,n})_{i,j} $ by the true covariances $ (\Sigma_{n,n})_{i,j} $.
That is, $\Sigma_n^\diamond = \E( \hat{\Sigma}_n^\diamond )$ and $\Sigma_n^\dagger = \E( \hat{\Sigma}_n^\dagger )$. 
Recall also that $ \Sigma_n = \E( \hat{\Sigma}_{n,n} )$. 
Denote for $A, B \ge 0 $ the scaled inner product by $ \langle A,B\rangle_* = A \cdot B / d $ and let $ \| A \|_{F*} = \sqrt{\langle A,A\rangle_*} $ be the associated scaled Frobenius matrix norm. This scaling ensures that the unit matrix has norm $1$ whatever the dimension $d$. The following theorem studies the  above estimators under the general time series model given by (P.1) and (P.2) and provides sufficient conditions on the growth of dimension $ d $ relative to the sample size when measuring the estimation error in terms of the scaled Frobenius risk. 

\begin{theorem} 
\label{ConsistencyTapperToeplitz}
	Suppose that $ \{ X_t \} $ satisfies (P.1) with $q\geq 4$ and $\beta>1$, and that $ \Cov(X_t) $ is bandable with exponent $\alpha>0$ uniformly for all $ t $, i.e.\ $\Cov(X_t)_{i,j} \leq C |i-j|^{-\alpha-1}$.
	
	\noindent(i)
	For any threshold $\tau$, it holds
    \[
      \E \| \hat{\Sigma}_n^\dagger - \ \Sigma_n \|_{F*}^2 
      \;\leq\; C(\Theta,\beta)\, \left\{ \tau^{-2\alpha-1} + \frac{\min(\tau,d)}{n} \right\}.
    \]
	Choosing $\tau = \min(n^{\frac{1}{2\alpha+2}},\, d)$, it holds
	\[
      \E \| \hat{\Sigma}_n - \ \Sigma_n^\dagger \|_{F*}^2 
      \;\leq\; C(\Theta,\beta)\, \min(n^{-\frac{2\alpha+1}{2\alpha+2}},\; \tfrac{d}{n}).
    \]

    \noindent(ii)
    If $\Sigma_n$ is a Toeplitz matrix, then 
    \[
     \E \| \hat{\Sigma}_n^\diamond - \ \Sigma_n \|_{F*}^2  
      \;\leq\; C(\Theta,\beta)\, \left\{ \tau^{-2\alpha-1} + \frac{\min(\tau,d)}{n} \right\}.
    \]
	Choosing $\tau = \min(n^{\frac{1}{2\alpha+2}},\, d)$, it holds
	\[
      \E \| \hat{\Sigma}_n^\diamond - \ \Sigma_n \|_{F*}^2  
      \;\leq\; C(\Theta,\beta)\, \min(n^{-\frac{2\alpha+1}{2\alpha+2}},\; \tfrac{d}{n}).
    \]
\end{theorem}

In particular, the optimal rate of convergence of the tapering estimator $\hat{\Sigma}_n^\dagger$ carries over to the dependent, nonstationary case, with identical threshold values.

A shrinkage estimator may be defined for any weight $w\in\Delta = \{(w_1, w_2,w_3)\in[0,1]^3 : w_1+w_2+w_3=1\}$, as 
\begin{align*}
	\hat{\Sigma}_{k,n}^w &= w_1 \hat{\Sigma}_{k,n} + w_2 \hat{\Sigma}_{k,n}^\dagger + w_3\hat{\Sigma}_{k,n}^\diamond.
\end{align*}
It turns out that the theory presented in Section \ref{sec:projection} may be applied to the quadratic form $v^T\hat{\Sigma}_n^w v$ for any shrinkage weight $w$, and hence also for the individual statistics $v^T\hat{\Sigma}_n v$, $v^T\hat{\Sigma}_n^\dagger v$, and $v^T\hat{\Sigma}_n^\diamond v$.
The crucial observation is that these quadratic forms can be regarded as linear projections of the $(d\times d)$-variate time series $X_t X_t^T$ with respect to the Frobenius inner product $ A \cdot B = \sum_{i,j} a_{ij} b_{ij} $ defined for $ d \times d $ matrices $ A = (a_{ij})_{i,j} $ and $ B = (b_{ij})_{i,j} $, which coincides with the usual inner product on $ \R^{d^2} $ of the vectorized matrices. Indeed, we have the representations
\[
  v^T \hat{\Sigma}_{k,n}^\dagger v = \nu^2 \cdot \hat{\Sigma}_{k,n}, \qquad v^T \hat{\Sigma}_{k,n}^\diamond v = \nu^3 \cdot \hat{\Sigma}_{k,n},
\]
for the projection weighting matrices $\nu^1,\nu^2, \nu^3\in\R^{d\times d}\equiv \R^{d^2}$ given by
\begin{align*}
	\nu^1 = (v_l v_{l'})_{l,l'=1}^d, &\qquad \|\nu^1\|_1 = \sum_{l,l'=1}^d |v_l v_{l'}| = \|v\|_1^2, \\
	\nu^2 = (v_l v_{l'} \omega(|l-l'|))_{l,l'=1}^d, & \qquad \|\nu^2\|_1 \leq \|\nu^1\|_1 = \|v\|_1^2, \\
	\nu^3 =  \left( \frac{\omega(|l-l'|)}{d-|l-l'|}\sum_{i-j=l-l'} v_iv_j\right)_{l,l'=1}^d, 
	& \qquad \|\nu^3\|_1 \leq \|\nu^1\|_1 = \|v\|_1^2
\end{align*}
In particular, we may write $v^T X_tX_t^T v = \nu^1 \cdot (X_t X_t^T)$ as the Frobenius inner product.

With Theorem \ref{thm:Gauss-ts-proj}, we obtain the following result.

\begin{theorem}\label{thm:Gauss-matrix}
	Suppose that the time series $X_t = G_t(\beps_t)$ satisfies \eqref{eqn:ass-ergodic-proj} resp.\ \eqref{eqn:ass-BV-proj} with power $2q$ and factor $\tilde{\Theta}$.
	Then, the time series $X_t$ may be defined on a different probability space, such that there exist independent, mean zero, Gaussian random vectors $\eta_t\in\R^3$, such that
	\begin{align*}
		\left( \E \max_{k=1,\ldots, n} \max_{w\in\Delta} \left|n v^T \left[n\hat{\Sigma}_{k,n}^w - \E \hat{\Sigma}_{k,n}^w\right] v - w^T \sum_{t=1}^k \eta_t\right|  \right)^\frac{1}{2} 
		\leq C \tilde{\Theta}^2 \|v\|_1^2 \sqrt{\log(n)}  n^{-\xi(q,\beta)}.
	\end{align*}
	The covariance of the Gaussian vectors $\eta_t\sim\mathcal{N}(0,\Xi_t)$ is given by
	\begin{align*}
		(\Xi_t)_{i,j} = \sum_{h=-\infty}^\infty \Cov\left( \nu^i \cdot G_t(\beps_0)G_t(\beps_0)^T,\; \nu^j \cdot G_t(\beps_h)G_t(\beps_h)^T \right),\quad i,j\in \{1,2,3\}.
	\end{align*}
\end{theorem}

It is worth mentioning that the bound in Theorem~\ref{thm:Gauss-matrix} allows that $ \| v \|_1 $ increases as the dimension, $d$, or the sample size, $n$, grow, as long as $ \| v \|_1 = o( (\log n)^{1/4} n^{-\xi(q,\beta)/2} ) $.

\section{Optimal shrinkage weights}\label{sec:optimal}

The question arises, how the shrinking weights relate to the theoretical (oracle) performance of the resulting shrinkage estimator. To pursue such a study, we first consider the problem to combine the nonparametric estimator with oracles of the two targets within a framework proposed by \cite{Ledoit2004} and aim at minimizing the squared Frobenius risk. This leads to a convex but box-constrained optimization problem allowing for an explicit interior solution (if it exists), which we briefly discuss. We propose consistent estimators of the unknown quantities determining the optimization problem and its solution(s) leading to data-adaptive bona fide estimates. 

For simplicity of presentation, we elaborate this for the full sample of size $n$; the corresponding sequential estimates and oracles based on the first $k$ observations can then be derived easily. Write $ \hat{\Sigma}_n = \hat{\Sigma}_{n,n} $ and recall that  $ \Sigma_n^\dagger $ and $ \Sigma_n^\diamond $ are the theoretical oracle estimators associated to the estimators $ \hat{\Sigma}_{n,n}^{\dagger} $ and $ \hat{\Sigma}_{n,n}^\diamond $, respectively, obtained by replacing in their definitions the estimates $ (\hat{\Sigma}_{n,n})_{i,j} $ by the true covariances $ (\Sigma_{n,n})_{i,j} $. Equivalently, $\Sigma_n^\diamond = \E \hat{\Sigma}_n^\diamond$ and $\Sigma_n^\dagger = \E \hat{\Sigma}_n^\dagger$.

We wish to determine the optimal ideal shrinking weights such that the resulting (scaled) Frobenius risk  is minimized when combining the nonparametric estimator  $ \hat{\Sigma}_n $ with the oracle estimators $ \Sigma_n^\dagger $ and $ \Sigma_n^\diamond $. The scaled mean squared error of the sample covariance matrix is given by $\operatorname{MSE}( \hat{\Sigma}_n ) = \operatorname{MSE}( \hat{\Sigma}_n; \Sigma_n ) = E \| \hat{\Sigma}_n - \Sigma_n \|_{F*}^2$, and we consider the risk minimization problem 
\begin{align}    
\label{RMinP}
    \begin{split}
  \min_w& \ \ \E \| \Sigma_n^w - \Sigma_n \|_{F*}^2 \\
  \text{where}\quad&\Sigma_n^w = w_1 \hat{\Sigma}_n + w_2 \Sigma_n^\dagger + w_3 \Sigma_n^\diamond,\qquad \ w_1+w_2+w_3=1.
  \end{split}
\end{align}
Further, let 
\[
  E_n^{\dagger} = \| \Sigma_n^\dagger - \Sigma_n \|_{F*}^2, \qquad  E_n^{\diamond} = \| \Sigma_n^\diamond - \Sigma_n \|_{F*}^2
\]
be the squared approximation errors corresponding to the oracles, and \[ D_n = \langle \Sigma_n^\dagger - \Sigma_n,  \Sigma_n^\diamond - \Sigma_n \rangle_*. \] 
Substituting $ w_1 = 1-w_2-w_3 $ and observing that $ \E \langle \hat{\Sigma}_n-\Sigma_n, A \rangle = 0$, for $ A \in \{ \Sigma_n^\dagger, \Sigma_n^\diamond \} $, it suffices to determine minimizers of
\begin{equation}
\label{OptSimplified}
  f( w_2, w_3 ) = (1-w_2-w_3)^2 \E \| \hat{\Sigma}_n - \Sigma_n \|_{F*}^2 + w_2^2 E_n^\dagger + w_3^2 E_n^\diamond + 2w_2w_3 D_n.
\end{equation}
Adding the constraints $ 0 \le w_1, w_2 \le 1 $ leads to a quadratic programming problem under box constraints, which generally requires numerical algorithms for its solution. We discuss the special case of an interior solution, where explicit formulas can be derived, at the end of this section. The optimal weights depend on the unknown values $MSE(\hat{\Sigma_n}) = \E \|\hat{\Sigma_n}-\Sigma_n\|_{F*}^2$, $E_n^\dagger$, $E_n^\diamond$, and $D_n$. To obtain a data-driven solution for the optimal weights, we need consistent estimators of these quantities, which are also interesting in their own right.

Observe that 
\[
  \operatorname{MSE}(\hat{\Sigma}_n) = \frac{1}{n d} \sum_{i,j=1}^d \Var( \sqrt{n} ( \hat{\Sigma}_{n} )_{i,j} ). 
\]
Under the assumptions of Proposition~\ref{prop:matrix}  $ \sqrt{n} (\hat{\Sigma}_{n} )_{i,j} = \frac{1}{\sqrt{n}} \sum_{t=1}^n (X_t)_i (X_t)_j $ can be approximated in $ L_2 $ by $ \frac{1}{\sqrt{n}} \sum_{t=1}^n Y_{t,i,j} $ for independent Gaussian random variables 
\[ Y_{t,i,j} \sim N( 0, \xi_{t,i,j}^2 ) \quad \text{with}\quad  \xi_{t,i,j}^2 = \sum_{h=-\infty}^\infty \Cov( G_{t,i}(\beps_0)G_{t,j}(\beps_0), \, G_{t,i}(\beps_h )G_{t,j}(\beps_h ) ). \]
This yields 
\[
\operatorname{MSE}(\hat{\Sigma}_n) = \frac{1}{n d} \sum_{i,j=1}^d \frac{1}{n} \sum_{t=1}^n \xi_{t,i,j}^2 + o(1).
\]

\cite{Sancetta2008} proposed an estimator of $ \operatorname{MSE}(\hat{\Sigma}_n) $ by estimating the long-run-variance parameters 
$
\Var( \sqrt{n} ( \hat{\Sigma}_{n} )_{i,j} ) 
$ by the long-run-variance estimators 
\[
  \widehat{\Var}(  \sqrt{n} ( \hat{\Sigma}_{n} )_{i,j} ) = \hat{\Gamma}_{n,i,j}(0) + 2 \sum_{s=1}^{n-1} K(s/b) \hat{\Gamma}_{n,i,j}(s),  
\]
where $ K $ is a positive weighting function, $b>0$ is a bandwidth, and
\[
    \hat{\Gamma}_{n,i,j}(s) =  \frac{1}{n} \sum_{t=1}^{n-s} [(X_t)_i (X_t)_j - (\hat{\Sigma}_n)_{i,j} ] \cdot [ (X_{t+s})_i (X_{t+s})_j - (\hat{\Sigma}_n)_{i,j}  ].
\]
The mean square error may then be estimated as 
\[
\widehat{MSE}(\hat{\Sigma}_n) = \frac{1}{nd} \sum_{i,j=1}^d \widehat{\Var}(\sqrt{n} (\hat{\Sigma}_n)_{i,j}).
\] 
The result is as follows, see \cite[Lemma~1]{Sancetta2008}.

\begin{lemma}\label{lem:MSE} If (P.1)-(P.3) hold, $ K : \R \to \R $ is a decreasing positive c\`adl\`ag-function with $\lim_{x\to0+} K(x) = 1$ and $ \int_0^\infty K^2(x) dx < \infty$, and $ b \to \infty $ with $ b = o( \sqrt{n} ) $, then
\[
  \left| \widehat{MSE}(\hat{\Sigma}_n) - MSE(\hat{\Sigma}_n) \right| = o\left( \frac{d}{n} \right).
\]
\end{lemma}
In practice, the bandwidth $b$ needs to be chosen in a data-driven way, e.g.\ by the method described in \cite{NeweyWest1994}.

Let us now derive consistent estimators for the approximation errors $ E_n^\dagger $ and $ E_n^\diamond $ and the inner product $D_n$. The quantity $ E_n^\diamond = \| \Sigma_n - \Sigma_n^\diamond \|_{F*}^2$ can be estimated by plugging in the estimators $ \hat{\Sigma}_n^\diamond $ and $\hat{\Sigma}_n^\dagger $. Hence, define
\[
  \hat{E}_n^\diamond = \| \hat{\Sigma}_n^\diamond - \hat{\Sigma}_n^\dagger \|_{F*}^2.
\]
To estimate $ E_n^\dagger $ observe that $ \Delta_n^\dagger  = ( \Delta_{n,ij}^\dagger )_{i,j} = \Sigma_n - \Sigma_n^\dagger $ is given by
\[
  \Delta_n^\dagger = \left\{
  \begin{array}{ll}
    0, & | i - j | \le \tau/2, \\
    (1-\omega(|i-j|)) \Sigma_{n,ij}, & \tau/2 < |i-j| \le \tau, \\
    \Sigma_{n,ij}, & |i-j| > \tau.
  \end{array}
  \right.
\]
Estimation of $\Delta_n^\dagger $ is delicate, since this involves estimating $\sum_{m>\tau} \sum_{|i-j|=m} 1 = (d-\tau)(d-\tau+1)/2 = O(d^2)$ off-diagonal terms which would lead to the unsatisfactory condition $ d^2 = o(n) $. For the class of bandable covariance matrices under consideration, we can proceed by introducing a second threshold and estimate only the covariances on the diagonals $ \tau/2+1, \ldots, \sigma $. This means, we estimate $ \Delta_n^\dagger $ by $ \hat{\Delta}_n^\dagger = ( \hat{\Delta}_{n,ij}^\dagger )_{i,j} $ with
\[
  \hat{\Delta}_{n,ij}^\dagger = \left\{
  \begin{array}{ll}
    0, & | i - j | \le \tau/2, \\
    (1-\omega(|i-j|)) \hat{\Sigma}_{n,ij}, & \tau/2 < |i-j| \le \tau, \\
    \hat{\Sigma}_{n,ij}, & \tau < |i-j| \le \sigma, \\
    0, & |i-j| > \sigma,
  \end{array}
  \right.
\]
for some threshold $ \sigma $ with $ \tau \ll \sigma \ll d $, and put
\[
  \hat{E}_n^\dagger = \| \hat{\Delta}_n^\dagger \|_{F*}^2.
\]
Lastly, the estimation of $ D_n = \langle \Sigma_n - \Sigma_n^\dagger, \Sigma_n - \Sigma_n^\diamond \rangle_* $ can be based on the above estimators of $ \Delta_n^\dagger = \Sigma_n - \Sigma_n^\dagger $ and $ \Sigma_n - \Sigma_n^\diamond $,
\[
  \hat{D}_n = \langle \hat{\Delta}_n^\dagger, \hat{\Sigma}_n^\dagger - \hat{\Sigma}_n^\diamond \rangle_*.
\]
Consistency of $ \hat{D}_n $, however, requires an extra condition compared to the estimators of the approximation errors. The reason is that $ \Sigma_n - \Sigma_n^\dagger $ tends to zero in the $\| \cdot \|_{F*} $-norm within the class of bandable matrices, whereas the error $ \Sigma_n - \Sigma_n^\diamond $ may diverge. But within the the subclass of bandable matrices which are located in a neighborhood of their Toeplitz-approximation in the sense that \[ \| \Sigma_n - \Sigma_n^\diamond \|_{F*} = O(1) \] consistency can be established, where $ \Sigma_n^\diamond $ is determined using the threshold $ n^{\frac{1}{2\alpha+c}} $; the threshold for the estimators are given below. We formulated the results for the highdimensional case $ d \ge n^{\frac{1}{2\alpha+c}} $.

\begin{theorem}
\label{ConsistencyApproxErrors}
	Suppose that $ \{ X_t \} $ satisfies (P.1) and (P.2) and $ \Var(X_t) $ is bandable for all $ t $.
Further, assume that $ d \ge n^{\frac{1}{2\alpha+c}} $, and the estimators $ \hat{\Sigma}_n^\dagger $, $ \hat{\Delta}_n^\dagger $ and $ \hat{\Sigma}_n^\diamond $ are calculated with thresholds $\tau_n^\dagger = n^{\frac{1}{2\alpha+c}} $, $\sigma_n^\dagger = n^{\frac{1+s}{2\alpha +c}} $ for some $ 0 < s < 2\alpha+c-1$ and $ \tau_n^\diamond = \left( \frac{nd}{\log nd} \right)^{\frac{1}{2\alpha +1} }$, respectively. 
\begin{itemize}
\item[(i)] If $ d = o( n^{\frac{2\alpha+2}{2\alpha + 1}} ) $, then
\[
 \E \| \hat{\Sigma}_n^\diamond - \hat{\Sigma}_n^\dagger  - (\Sigma_n^\diamond - \Sigma_n ) \|_{F*}  = o(1).
\]
and
\[
  \E \left| 
    \| \hat{\Sigma}_n^\diamond - \hat{\Sigma}_n^\dagger \|_{F*} - \| \Sigma_n^\diamond - \Sigma_n \|_{F*} \right| = o(1),
\]
which implies $ \| \hat{E}_n^\diamond - E_n^\diamond \|_{F*} \to 0 $ in probability.

\item[(ii)] We have
\[
  \E \| \hat{\Delta}_n^\dagger - \Delta_n^\dagger \|_{F*}^2 = 
  \left\{
  \begin{array}{ll}
    O\left( n^{-\frac{2\alpha-1+c-s}{2\alpha+c}} \right), & \qquad s > \frac{c}{2\alpha}, \\
    O\left( n^{-\frac{(s+1)(2\alpha-1)}{2\alpha+c}} \right), & \qquad s \le \frac{c}{2\alpha}, \\    
  \end{array}
  \right.
\]
and
\[
  \E \left|  \| \hat{\Delta}_n^\dagger \|_{F*} - \| \Delta_n^\dagger \|_{F*} 
    \right| = o(1),
\]
such that $ \| \hat{E}_n^\dagger - E_n^\dagger \|_{F*} \to 0 $ in probability.
\item[(iii)] If $  \| \Sigma_n - \Sigma_n^\diamond \|_{F*} = O(1) $ and $ d = o( n^{\frac{2\alpha+2}{2\alpha + 1}} ) $, then
\[
  | \hat{D}_n - D_n | \to 0,
\]
in probability.
\end{itemize}
\end{theorem}

If, for example, $ \alpha = 1 $, then the dimension needs to satisfy $ d^{3/4}=o(n)$, i.e. it may grow almost as fast as $ n^{4/3} $, and for $ \alpha = 2 $ almost as fast as $ n^{6/5} $.

Having consistent estimators at our disposal, we can propose the following bona fide estimators of the optimal weights, 
\[ 
    \left( \hat{w}_1^*, \hat{w}_2^*, \hat{w}_3^*   \right) = \arg\min_{w\in \Delta} w_1^2 \widehat{MSE}(\hat{\Sigma}_n) + w_2^2 \hat{E}_n^\dagger + w_3^2 \hat{E}_n^\diamond + 2 w_2 w_3 \hat{D}_n.
\]
This yields the feasible shrinkage estimator
\begin{align*}
    \hat{\Sigma}_n^{w^*} = w_1^* \hat{\Sigma}_n + w_2^* \hat{\Sigma}_n^\dagger + w_3^* \hat{\Sigma}_n^\diamond.
\end{align*}
The performance of the corresponding shrinkage estimator $\hat{\Sigma}_n^{\hat{w}^*}$ is assessed via simulations in the Section \ref{sec:simulation}.

Let us now briefly discuss the case of an interior solution of the minimization problem \eqref{OptSimplified}. 
We have
\begin{align*}
  f( w_2, w_3 ) &= ( w_2,w_3) Q_n \begin{pmatrix} w_2 \\ w_3 \end{pmatrix} - a_n^\top \begin{pmatrix} w_2 \\ w_3 \end{pmatrix} + \operatorname{MSE}(\hat{\Sigma}_n),
  \end{align*}
where
\[
Q_n = \begin{pmatrix} \operatorname{MSE}(\hat{\Sigma}_n) + E^\dagger & \operatorname{MSE}(\hat{\Sigma}_n) + D_n \\ \operatorname{MSE}(\hat{\Sigma}_n) + D_n & \operatorname{MSE}(\hat{\Sigma}_n) + E_n^\diamond \end{pmatrix},
\quad a_n = \begin{pmatrix} \operatorname{MSE}(\hat{\Sigma}_n) \\ \operatorname{MSE}(\hat{\Sigma}_n) \end{pmatrix}.
\]
Hence, we are given a quadratic minimization problem under box constraints, and a natural condition is to assume that $ \operatorname{det}(Q_n) > 0$. If an interior solution $ w^* \in (0,1)^3$ exists (or the box constraint is omitted), the explicit solution is easily seen to be given by 
\[
  w_2^* = \frac{ \operatorname{MSE}(\hat{\Sigma}_n)( E_n^\diamond - D_n) }{ F_n }, \qquad w_3^* = \frac{\operatorname{MSE}(\hat{\Sigma}_n)( E_n^\dagger - D_n )}{F_n}
\]
with
\[
  F_n = \operatorname{det}(Q_n) = E_n^\dagger E_n^\diamond + \operatorname{MSE}( \hat{\Sigma}_n )( E_n^\dagger + E_n^\diamond - 2D_n ) - D_n^2.
\]
The resulting optimal weight for the sample covariance matrix $ \hat{\Sigma}_n $ is
\[
  w_1^* = \frac{E_n^\dagger E_n^\diamond - D_n^2}{F_n}. 
\]
We see from the formula for $ w_1^* $ that the optimal solution prefers the nonparametric estimator $ \hat{\Sigma}_n $, if the truth $ \Sigma_n $ cannot be well approximated by a banded or Toeplitz matrix. The banded estimator receives a large weight, if the distance between the Toeplitz oracle and $ \Sigma_n $ is large. Analogously, the optimal solution assigns a large weight to the Toeplitz estimator, if $ \Sigma_n $ is not well approximated by the banded oracle, so that $ \| \Sigma_n - \Sigma_n^\dagger \|_F^2 $ is large.  As a result, it is optimal to distribute the weights across the estimators $ \hat{\Sigma}_n, \hat{\Sigma}_n^\dagger $ and $ \hat{\Sigma}_n^\diamond $ according to how well they approximate $ \Sigma_n $, where the approximation accuracy is measured by the squared distances of the oracles to the truth and, in case of the nonparametric estimator, by the MSE. 

Since $ D_n^2 \le E_n^\dagger E_n^\diamond $, we have $ w_1^* \ge 0 $. If $D_n \le \min( E_n^\dagger, E_n^\diamond) $, then $ w_2^*, w_3^* \ge 0 $ so that $ w \in \Delta $. 
It is interesting to note that $ \frac{w_2^*}{w_3^*} = \frac{E_n^\diamond - D_n}{E_n^\dagger - D_n} $, i.e., the ratio of the oracle shrinking weights no longer dependents on $ \operatorname{MSE}(\hat{\Sigma}_n) $. Further, $ w_{2}^* > w_3^* $ (banding preferable compared to Toeplitz) if and only if 
\[ E_n^\dagger < E_n^\diamond \Leftrightarrow \| \Sigma_n^\dagger - \Sigma_n \|_{F*} < \| \Sigma_n^\dagger - \Sigma_n \|_{F*}. \]
Even more is true: The above derivations remain formally true without any changes, if we use an arbitrary unbiased estimator $ \tilde{\Sigma}_n $ of $ \Sigma_n $, i.e., an arbitrary measurable function $ \tilde{\Sigma}_n $ of the data with $ \E( \tilde{\Sigma}_n ) = \Sigma_n $. Therefore, we have the following interesting result: The ratio $\frac{w_2^*}{w_3^*} $ is a universal characteristic of the risk minimization problem (\ref{RMinP}) in the sense that it is independent of the nonparametric estimator and completely determined by $ \Sigma_n $.

%

\section{Testing for structural breaks}\label{sec:changepoint}

The sequential approximation result of Theorem \ref{thm:Gauss-matrix} may be applied to test for changes in the covariance matrix $\Sigma_t = \Cov(X_t, X_t)\in\R^{d\times d}$ of a time series $X_t$.
To be precise, we want to test the hypothesis
\begin{align*}
	H_0: \Sigma_t \equiv \Sigma_0 \quad\leftrightarrow\quad H_1: \Sigma_t \not\equiv \Sigma_0.
\end{align*}
We propose the following CUSUM-type statistics
\begin{align*}
	T_n(w) &= \sqrt{n} \max_{k\leq n}  \left|v^T \left(\hat{\Sigma}^w_{k,n} - \tfrac{k}{n} \hat{\Sigma}^w_{n,n}\right) v \right|,\quad w\in\Delta, \\
	T_n^*(D) &= \sup_{w\in D}\; T_n(w), \qquad D \subset \Delta.
\end{align*}

To determine critical values for the statistic $T_n(w)$ resp.\ $T_n^*$, we employ a bootstrap scheme suggested in \cite{mies2022}.
In particular, let $\eta_t\sim \mathcal{N}(0,A_t)$ be independent Gaussian random vectors in $\R^3$, where
\begin{align*}
	A_t 	&= \frac{1}{b} \left(\sum_{s=t-b+1}^t \chi_s\right)^{\otimes 2} , \qquad t=b,\ldots, n, \\
	\chi_s	&= \left[\nu^j \cdot (X_s X_s^T - \hat{\Sigma}_{n,n}) \right]_{j=1}^3.
\end{align*}
and $A_t=0$ for $t<b$. 
Here, $b$ is a block-length tuning parameter which needs to be chosen suitably, see below.
We now define the bootstrapped version of the changepoint statistics as
\begin{align*}
	R_n(w) &= \sqrt{n}\max_{k\leq n} \left|  \frac{1}{n}\sum_{t=b}^k w^T\eta_t -  \frac{k}{n} \sum_{t=b}^n w^T \eta_t  \right|, \\
	R_n^*(D)&= \sup_{w\in D} R_n(w).
\end{align*}
Finally, for $\alpha\in(0,1)$, introduce the conditional quantiles
\begin{align*}
	a_n(\alpha,w) &= \inf \{a>0\; :\;  P(R_n(w)|X_1,\ldots, X_n) \leq \alpha \}, \\
	a_n^*(\alpha,D) &= \inf \{a>0\; :\;  P(R_n^*(D)|X_1,\ldots, X_n) \leq \alpha \}.
\end{align*}
We suggest to reject the null hypothesis if
\begin{align*}
	T_n(w) > a_n(\alpha,w)+\delta_n, \qquad \text{resp.} \qquad T_n^*(D) > a_n^*(\alpha_n,D)+\delta_n,
\end{align*}
for $\delta_n = 1/(\log n)^p$ for some $p\geq 1$.

\begin{theorem}\label{thm:bootstrap}
	Let $X_{t} = X_{t,n} = G_{t}^n(\beps_t)$ be an array of $d_n$-variate time series, such that each kernel $G_{t}^n$ satisfies (P.1) and (P.2), for some $q>8$, $\beta>2$, and with factors $\tilde{\Theta}$ and $\tilde{\Gamma}$ not depending on $n$.
	Choose the block length $b$ such that $b=b_n \asymp n^\zeta$ for some $\zeta\in(0,1)$.
	If $\Cov(X_{t,n}) = \Sigma_0$ for all $t=1,\ldots, n$, and for some covariance matrix $\Sigma_0$, then
	\begin{align*}
		\limsup_{n\to\infty} P\left(  T_n(w) > a_n(\alpha,w) + \delta_n  \right) &\leq \alpha, \qquad w\in\Delta,\\
		\limsup_{n\to\infty} P\left(  T_n^*(D) > a_n^*(\alpha,D) + \delta_n  \right) &\leq \alpha, \qquad D\subset\Delta.
	\end{align*}
\end{theorem}

Hence, the proposed bootstrap scheme indeed maintains the specified size for a change in covariance. 
A particular feature of this test is that it is robust against nuisance changes, because even under the null hypothesis the time series $X_t$ may be nonstationary, except for its stationary marginal covariance matrix. 
The relevance of this kind of robustness has been first highlighted by \cite{Zhou2013} for changes in mean, see also \cite{Gorecki2018} and \cite{Pesta2018}. Changes in the second moment structure in the presence of non-constant mean have been studied by \cite{Dette2018} and \cite{schmidt2021}, and a robust CUSUM test for a broad class of parameters is introduced in \cite{mies2021}. Nevertheless, all these references focus on the low-dimensional case, whereas our new test is also applicable in high dimensions.

The power of the proposed test and bootstrap scheme is hard to investigate analytically, because the model framework is rather broad. Instead, we analyze the behavior of the proposed changepoint test via simulations in the subsequent section.

\section{Simulation}\label{sec:simulation}

To demonstrate the implementation of our proposed changepoint test and the benefits of shrinkage, we assess our methodology for three examples of high-dimensional time series.

\subsection{Model A}
First, we simulate a high-dimensional time series according to the vector ARMA(1,1) model
\begin{align*}
	X_t = a \Pi X_{t-1} + b\epsilon_t + \epsilon_{t-1}.
\end{align*}
The matrix $\Pi\in\R^{d\times d}$ is a permutation matrix, such that the model is non-explosive.
Furthermore, for any $t$, the $(\epsilon_t)_i$, $i=1,\ldots, d$, are chosen as independent, zero mean random variables having a symmetrized Gamma distribution with shape parameter $\alpha(t/n)$, standardized to unit variance. 
We set $\alpha(u) = 2+\sin(2\pi u)$, $u\in[0,1]$. 
Thus, the autocovariance structure is stationary, but the time-series is nonstationary. 
In particular, the long run covariance matrices $\Xi_t$ which occur in Theorem \ref{thm:Gauss-matrix} are non-constant.

We assess our changepoint procedure under the null hypothesis of no change, and under the alternative where $b$ changes at time $t= \lceil\frac{n}{2}\rceil$.
The projection vector $v$ is chosen randomly as $N/\|N\|_1$ for a standard Gaussian random vector $N\sim\mathcal{N}(0,I_{d\times d})$.
The threshold for the tapering estimator is chosen as $\tau^{\dagger}_n = n^{\frac{1}{2\alpha+1}} $, and $\tau_n^{\diamond} = (nd/\log(nd))^\frac{1}{2\alpha+1}$. 
Note that these thresholds are optimal for estimation under the spectral norm, in the class of bandable matrices with decay rate $\alpha$. 
In practice, $\alpha$ is unknown, and we use the threshold for $\alpha=2$ in our simulations.
To determine the shrinkage weights, we also use the threshold $\sigma=\sigma_n^\dagger = n^{\frac{1}{2\alpha+1}} \gg \tau_n^{\dagger}$.
For the bootstrap scheme, the additional offset is chosen as $\delta_n = \log(n)^{-4}$, and the block-length is $b_n=\lceil 5n^{0.2} \rceil$.

Table \ref{tab:power} reports the simulated size and power of the bootstrap test for different combinations of $n$ and $d=d_n$, and shrinkage weights $w^{(1)} = (1,0,0)$ and $w^{(2)}=(0.3, 0.3, 0.4)$.
We also consider the data-driven shrinkage weight $w^*$ as described in Section \ref{sec:optimal}, which is not covered by our Gaussian approximation results.
Regarding the choice of the bandwidth for the estimator $\widehat{MSE}(\hat{\Sigma}_n)$, we employ the method of \cite{NeweyWest1994} with Bartlett kernel as implemented in the R package \textsc{sandwich} (see \cite{sandwich2020}), and we determine a separate bandwidth for each component $\widehat{\Var}(\sqrt{n}(\hat{\Sigma}_n)_{i,j})$.

When simulating the process, we set $a=0.5$ and $b=0.5$ (resp.\ $b=0.75$ post-change).
Based on our simulation results, we find that the test maintains the specified size $10\%$ and is indeed slightly conservative, as established theoretically in Section \ref{sec:changepoint}. 
Moreover, the power of the changepoint test increases with sample size, as expected, but also with increasing dimension. 
This may be explained by the fact that here, the change affects all coordinates, such that the additional information can be used to improve the detection performance. 
It is also interesting to observe that the additional usage of the tapered and the Toeplitz estimator further improves the power of the test.
This is interesting because the shrinkage estimator has originally been proposed as a method to improve the performance of a point estimate. 
Our simulation results demonstrate that shrinkage can also be also useful for testing.

\begin{table}
    \centering
    \begin{tabular}{ll|cccc}
         &  & $n=100$ & $n=500$ & $n=1000$ & $n=2000$ \\ \hline\hline
         && \multicolumn{4}{c}{Model A} \\ \hline\hline
         $w^{(1)}$    & $d=4$    & 0.102 (0.09) & 0.220 (0.08) & 0.433 (0.09) & 0.727 (0.09) \\ 
         & $d=4n^{0.3}$   & 0.101 (0.09) & 0.289 (0.08) & 0.520 (0.08) & 0.824 (0.09) \\ 
         & $d=4n^{0.5}$  & 0.098 (0.07) & 0.282 (0.08) & 0.535 (0.08) & 0.822 (0.09) \\ \hline
         $w^{(2)}$   & $d=4$  & 0.107 (0.08) & 0.334 (0.08) & 0.594 (0.09) & 0.878 (0.10) \\ 
         & $d=4n^{0.3}$   & 0.155 (0.07) & 0.823 (0.08) & 0.990 (0.09) & 1.000 (0.09) \\ 
         & $d=4n^{0.5}$  & 0.201 (0.05) & 0.967 (0.07) & 1.000 (0.07) & 1.000 (0.08) \\ \hline
         $w^{*}$ & $d=4$  & 0.102 (0.08) & 0.252 (0.08) & 0.458 (0.09) & 0.751 (0.09) \\ 
         & $d=4n^{0.3}$   & 0.119 (0.08) & 0.530 (0.08) & 0.826 (0.08) & 0.985 (0.09) \\ 
         & $d=4n^{0.5}$  & 0.130 (0.06) & 0.697 (0.07) & 0.956 (0.07) & 1.000 (0.09) \\  \hline\hline
         && \multicolumn{4}{c}{Model B} \\ \hline\hline
    $w^{(1)}$    & $d=4$  & 0.108 (0.08) & 0.277 (0.08) & 0.503 (0.09) & 0.803 (0.09) \\ 
        & $d=4n^{0.3}$    & 0.102 (0.08) & 0.275 (0.07) & 0.502 (0.08) & 0.804 (0.09) \\ 
        & $d=4n^{0.5}$    & 0.101 (0.08) & 0.276 (0.07) & 0.505 (0.08) & 0.798 (0.09) \\ \hline
    $w^{(2)}$   & $d=4$   & 0.137 (0.08) & 0.438 (0.08) & 0.712 (0.09) & 0.938 (0.09) \\ 
        & $d=4n^{0.3}$    & 0.198 (0.08) & 0.861 (0.07) & 0.993 (0.08) & 1.000 (0.09) \\ 
        & $d=4n^{0.5}$    & 0.270 (0.07) & 0.966 (0.08) & 1.000 (0.08) & 1.000 (0.09) \\ \hline
        $w^{*}$ & $d=4$ & 0.115 (0.08) & 0.302 (0.08) & 0.531 (0.09) & 0.829 (0.09) \\ 
        & $d=4n^{0.3}$    & 0.133 (0.08) & 0.528 (0.07) & 0.837 (0.08) & 0.989 (0.09) \\ 
        & $d=4n^{0.5}$    & 0.157 (0.07) & 0.671 (0.07) & 0.941 (0.08) & 0.999 (0.09) \\
        \hline\hline
        && \multicolumn{4}{c}{Model C} \\ \hline\hline
        $w^{(1)}$    & $d=4$  & 0.098 (0.08) & 0.273 (0.08) & 0.502 (0.08) & 0.806 (0.09) \\ 
        & $d=4n^{0.3}$        & 0.104 (0.08) & 0.277 (0.08) & 0.510 (0.08) & 0.806 (0.09) \\ 
        & $d=4n^{0.5}$        & 0.094 (0.08) & 0.269 (0.07) & 0.506 (0.08) & 0.799 (0.09) \\ \hline
    $w^{(2)}$   & $d=4$       & 0.124 (0.08) & 0.415 (0.08) & 0.692 (0.09) & 0.929 (0.09) \\ 
        & $d=4n^{0.3}$ & 0.170 (0.08) & 0.774 (0.08) & 0.971 (0.08) & 1.000 (0.09) \\ 
        & $d=4n^{0.5}$ & 0.225 (0.07) & 0.908 (0.08) & 0.997 (0.08) & 1.000 (0.09) \\ 
        \hline
     $w^{*}$   & $d=4$ & 0.108 (0.08) & 0.321 (0.08) & 0.573 (0.08) & 0.856 (0.09) \\ 
        & $d=4n^{0.3}$ & 0.132 (0.08) & 0.460 (0.08) & 0.733 (0.08) & 0.947 (0.09) \\ 
        & $d=4n^{0.5}$ & 0.143 (0.08) & 0.518 (0.07) & 0.786 (0.08) & 0.952 (0.09) \\ \hline
    \end{tabular}
	\caption{Power (and size), nominal level $10\%$. p-values based on $10^3$ bootstrap samples, reported values based on $10^4$ Monte Carlo simulations.}
    \label{tab:power}
\end{table}

\subsection{Models B \& C}
As a second example, we consider the vector ARMA(1,1) process $X_t$ given by
\begin{align*}
    X_t = a X_{t-1} + b\epsilon_t + \epsilon_{t-1},
\end{align*}
with $a$ and $b$ as above, and $\epsilon_t\sim \mathcal{N}(0, A)$ iid random vectors, and we initialize $X_0 \sim \mathcal{N}(0,A\,(1+b^2)/(1-a^2))$ to ensure weak stationarity. 
The symmetric positive semidefinite matrix $A\in\R^{d\times d}$ is chosen as $A_{i,j}=\gamma(t_i-t_j)$, for $\gamma(h) = |h+1|^{2H} + |h-1|^{2H} - 2|h|^{2H}$, which is the autocovariance function of fractional Gaussian noise. y
For the values $t_i$, we consider the scenarios (B) $t_i=i$, such that $A$ is a Toeplitz matrix, and (C) $t_i = \sqrt{i}$, such that $A$ is not a Toeplitz matrix.
For the bootstrap scheme and for the determination of shrinkage weights, we use the same settings as in Model A.

The size and power of the changepoint test are presented in Table \ref{tab:power}.
As for Model (A), the test turns out to be conservative, and gains power with increasing sample size, and higher dimension. 
In our simulations, we also find the deterministic weight $w^{(2)}$ leading to higher power compared to the data-driven weight $w^*$.
Thus, future work could explore the optimal choice of shrinkage weights from a testing perspective.

Under the null hypothesis of no change, the marginal covariance matrix of $X_t$ is given by $\Cov(X_t) = A \frac{1+2ab + b^2}{1-a^2}$. 
Thanks to this explicit formula, we are able to analyze the error of estimation of the proposed shrinkage estimator with data-driven shrinkage weights, see Table \ref{tab:error-ARMA}. 
As benchmarks, we evaluate the performance of the estimators $\hat{\Sigma}_n$, $\hat{\Sigma}_n^\dagger$, and $\hat{\Sigma}_n^\diamond$ individually. 
For model B, the Toeplitz estimator $\hat{\Sigma}_n^\diamond$ is found to perform best, which could be expected because the true covariance matrix is of Toeplitz type. 
On the other hand, the shrinkage estimator always improves upon the sample covariance matrix $\hat{\Sigma}_n$, and is often close the performance of the tapered estimator $\hat{\Sigma}_n^\dagger$.
For model C, the decay of the off-diagonal entries of the true covariance matrix is slower, hence the performance of $\hat{\Sigma}_n^\dagger$ is also worse. 
Indeed, its estimation error is larger than the error of the sample covariance $\hat{\Sigma}_n$.
Yet, for $n\geq 500$, the error of the shrinkage estimator is lower than that of all other estimators. 
This demonstrates that our proposed shrinkage estimator not only chooses among the three estimators, but may indeed improve the performance even further due to the convex combination.

\begin{table}
    \centering
    \small
    \begin{tabular}{l|rrrr}
         & \multicolumn{1}{c}{$n=100$} & \multicolumn{1}{c}{$n=500$} & \multicolumn{1}{c}{$n=1000$} & \multicolumn{1}{c}{$n=2000$} \\ 
         \hline\hline
         & \multicolumn{4}{c}{Model B} \\ \hline\hline
         $d=4$ & 2.50$|$2.32$|$0.87$|$2.34 & 0.52$|$0.48$|$0.18$|$0.48 & 0.26$|$0.24$|$0.09$|$0.24 & 0.13$|$0.12$|$0.05$|$0.12 \\ 
        $d=4n^{0.3}$ & 6.38$|$3.84$|$0.78$|$4.01 & 2.55$|$1.29$|$0.16$|$1.42 & 1.49$|$0.76$|$0.08$|$0.82 & 0.95$|$0.49$|$0.04$|$0.53 \\ 
        $d=4n^{0.5}$  & 19.96$|$4.32$|$0.60$|$7.71 & 9.06$|$1.43$|$0.15$|$3.17 & 6.39$|$0.86$|$0.08$|$2.13 & 4.54$|$0.56$|$0.04$|$1.48 \\ 
        \hline\hline
        & \multicolumn{4}{c}{Model C} \\ \hline\hline
        $d=4$ & 2.63$|$2.61$|$1.49$|$2.45 & 0.55$|$0.67$|$0.61$|$0.52 & 0.28$|$0.42$|$0.50$|$0.27 & 0.14$|$0.29$|$0.44$|$0.14 \\ 
        $d=4n^{0.3}$ & 6.85$|$5.17$|$3.91$|$4.80 & 2.73$|$2.55$|$5.21$|$1.98 & 1.58$|$1.79$|$5.69$|$1.23 & 1.01$|$1.26$|$6.80$|$0.82 \\ 
        $d=4n^{0.5}$ & 21.25$|$12.47$|$9.05$|$10.64 & 9.48$|$11.13$|$12.98$|$5.75 & 6.66$|$12.37$|$15.69$|$4.55 & 4.70$|$12.27$|$19.12$|$3.53 \\ \hline
    \end{tabular}
    \caption{Mean square estimation error of the marginal covariance $\Sigma=\Cov(X_t)$ under stationarity, measured in the scaled Frobenius norm $\|\cdot\|_{F*}$. Errors are reported, in this order, for (i) sample covariance $\hat{\Sigma}_n$, (ii) tapered estimator $\hat{\Sigma}_n^\dagger$, (iii) Toeplitz-type estimator $\hat{\Sigma}_n^\diamond$, (iv) the shrinkage estimator $\hat{\Sigma}_n^{w^*}$. Reported values based on $10^4$ Monte Carlo simulations.}
    \label{tab:error-ARMA}
\end{table}

\section{Proofs}\label{sec:proofs}

\begin{proof}[Proof of Lemma \ref{lem:jensen}] The result follows from Jensen's inequality:
\begin{align*}
	\E \left| \sum_{j=1}^d v_j X_{tj} \right|^{q} 
	\le \| v \|_{1}^{q} \E \left( \sum_{j=1}^d \frac{ |v_j| }{\| v \|_{1}} | X_{tj} | \right)^{q} 
	\le \| v \|_{1}^{q}  \sum_{j=1}^d \frac{ |v_j| }{\| v \|_{1}} \E | X_{tj} |^{q} 
	 = \| v \|_{1}^{q} \max_{1 \le j \le d} \E | X_{tj} |^{q}
\end{align*}
\end{proof}

\begin{proof}[Proof of Theorem \ref{thm:Gauss-ts-proj}]
	The time series $Z_t$ may be written as $Z_t = \tilde{G}_t(\beps_t)$, with kernel $\tilde{G}_t = V G_t$, such that we may apply Theorem 3.1 of \cite{mies2022}.
	It can be verified that $\tilde{G}_t$ satisfies conditions (G.1) and (G.2) therein, with $\Gamma=\tilde{\Gamma}$ and $\Theta = \sqrt{m}\tilde{\Theta} \|V\|_\infty$.
	To see this, note that for any $m$-variate random variable $X$, we have
	\begin{align*}
		(\E \|X\|_2^q)^\frac{1}{q} 
		\leq (\E \|X\|_q^q)^\frac{1}{q} m^{\frac{1}{2}-\frac{1}{q}} 
		\leq m^\frac{1}{2} \max_{l=1,\ldots,m} (\E |X_l|^q)^\frac{1}{q},
	\end{align*}
	because $\|x\|_2 \leq m^{\frac{1}{2}-\frac{1}{q}} \|x\|_q$ for any vector $x\in\R^m$.
	Moreover, for any $l=1,\ldots,m$,
	\begin{align*}
		\left( \E |V_{l,\cdot} G_t(\beps_0)|^q\right)^\frac{1}{q} 
		&= \left( \E \left|\sum_{r=1}^d V_{l,r} G_{t,r}(\beps_0)\right|^q\right)^\frac{1}{q} \\
		&\leq  \sum_{r=1}^d |V_{l,r}| \, \left(\E |G_{t,r}(\beps_0)|^q\right)^\frac{1}{q} \qquad
		\leq \|V\|_\infty \tilde{\Theta},
	\end{align*}
	such that
	\begin{align*}
		\left( \E \|V G_t(\beps_0)\|^q\right)^\frac{1}{q} \leq \sqrt{m} \|V\|_\infty \tilde{\Theta}.
	\end{align*}
	The same argument can be used to establish (G.1) and (G.2).
\end{proof}

\begin{proof}[Proof of Proposition \ref{prop:matrix}]
	Observe that for four random variables $Y_1, Y_2, Z_1, Z_2 \in L_{2q}$, it holds that
	\begin{align*}
		(\E |Y_1Z_1 - Y_2 Z_2|^{q})^\frac{1}{q} 
		&\leq (\E |(Y_1-Y_2)Z_1|^{q})^\frac{1}{q} + (\E |Y_2(Z_1- Z_2)|^{q})^\frac{1}{q} \\
		&\leq (\E |Y_1-Y_2|^{2q})^\frac{1}{2q} (\E |Z_1|^{2q})^\frac{1}{2q} + (\E |Y_w|^{2q})^\frac{1}{2q}(\E |Z_1-Z_2|^{2q})^\frac{1}{2q}.
	\end{align*}
	Inequality \eqref{eqn:ass-ergodic-proj} may be obtained by setting $Y_1 = G_{t,l}(\beps_t), Y_2 = G_{t,l}(\tilde{\beps}_{t,t-j})$ and $Z_1 = G_{t,l'}(\beps_t), Z_2 = G_{t,l'}(\tilde{\beps}_{t,t-j})$.
	Inequality \eqref{eqn:ass-BV-proj} may be obtained by setting $Y_1 = G_{t,l}(\beps_t), Y_2 = G_{t-1,l}(\beps_{t})$ and $Z_1 = G_{t,l'}(\beps_t)$, $Z_2 = G_{t-1,l'}(\beps_{t})$. 
\end{proof}

\begin{proof}[Proof of Theorem \ref{ConsistencyTapperToeplitz}]
	 By \cite[Th.~3.2]{mies2022} we have under Assumption (P.1) for all $d, n$
	 \[
	  \left( \E \max_{k \le n} \left| \sum_{t=1}^k X_{ti} X_{tj} - ( \Var(X_t) )_{ij} \right|^2 \right)^{\frac{1}{2}} 
	  \leq C(\Theta,\beta) n^\frac{1}{2},\quad 1\leq i,j\leq d,
	 \]
	 such that $ \hat{\Sigma}_{n,ij} = ( \hat{\Sigma}_n)_{ij} $ satisfies 
	 \begin{equation}
	    \max_{1 \le i, j \le d} \E( \hat{\Sigma}_{n,ij} -  \ \Sigma_{n,ij} )^2 \leq C(\Theta,\beta) n^{-1}.
	 \end{equation}
	 
	 \noindent\underline{To show the bound for $\hat{\Sigma}_n^\dagger$}, observe that
	 \[
	   \E \| \hat{\Sigma}_n^\dagger - \ \Sigma_n \|_{F*}^2 \leq 4\, ( A_n + B_n + C_n )
	 \]
	 where
	 \begin{align*}
	 	  A_n & = d^{-1} \sum_{m \le \tau/2} \sum_{|i-j|=m} \E( \hat{\Sigma}_{n,ij} -  \ \Sigma_{n,ij} )^2, \\
	 	  B_n &= d^{-1} \sum_{\tau/2 < m \le \tau} \sum_{|i-j|=m} \omega(|i-j|)^2 \E\left(  \hat{\Sigma}_{n,ij} -   \Sigma_{n,ij} \right)^2, \\
	 	  C_n &= d^{-1} \sum_{\tau/2 < m \le \tau} \sum_{|i-j|=m}  (1-\omega(|i-j|))^2 \Sigma_{n,ij}^2 , \\
	 	  D_n & =  d^{-1} \sum_{m > \tau} \sum_{|i-j|=m} \Sigma_{n,ij}^2.
 	\end{align*}
    Clearly, $ |C_n + D_n| = O( \tau^{-2\alpha-1} ) $. 
    Furthermore, 
    \begin{align*}
     |A_n + B_n| 
     \leq  \frac{C}{n d} \sum_{m\leq \tau} \sum_{|i-j|=m} 1 
     \leq \frac{C}{n} \min(\tau,d).
     \end{align*}
    Thus, 
    \begin{align*}
        \E \| \hat{\Sigma}_n^\dagger - \ \Sigma_n \|_{F*}^2
        \leq C \left\{ \tau^{-2\alpha-1} + \frac{\min(\tau,d)}{n} \right\}.
    \end{align*}
    The optimal rate is attained for $\tau = \min( n^{\frac{1}{2\alpha+2}},\, d)$. 
   
   \noindent\underline{To show the bound for $\hat{\Sigma}_n^\diamond$} observe that by Jensen's inequality
   \begin{align*}
     \E( \hat{\sigma}_m - \ \sigma_m )^2 
     & =\E \left(  \frac{1}{d-m} \sum_{k-l=m} \hat{\Sigma}_{n,kl} - \ \Sigma_{n,kl} \right)^2 \\
     & \le   \frac{1}{d-m}  \sum_{k-l=m}  \E( \hat{\Sigma}_{n,kl} - \ \Sigma_{n,kl} )^2
     \\
     & \leq C(\Theta,\beta) n^{-1},
   \end{align*}
where $ \ \sigma_m = \frac{1}{d-m} \sum_{k-l=m} \ \Sigma_{n,kl} $ for all $i,j $ with $ i-j= m \le \tau$. We have
\begin{align*}
	\E \| \hat{\Sigma}_n^\diamond - \ \Sigma_n \|_{F,*}^2 
	& \le \frac{2}{d} \sum_{m \le \tau} \sum_{i-j=m} \E(  \hat{\Sigma}_{n,ij}^\diamond - \ {\Sigma}_{n,ij} )^2 + \frac{2}{d} \sum_{m>\tau} \sum_{i-j=m} \ ({\Sigma}_{n,ij})^2 \\
	& = \frac{2}{d} \sum_{m \le \tau} \sum_{i-j=m} \E( \omega(i-j) ( \hat{\sigma}_m -  \ \sigma_m ) )^2 + \frac{2}{d} \sum_{m>\tau} \sum_{i-j=m} \ ({\Sigma}_{n,ij})^2 \\
	& = G_n + H_n.
\end{align*}
Since $ \ {\Sigma}_{n,ij}^\diamond \leq C \, m^{-\alpha-1}$, we have $\sigma_m \leq C m^{-\alpha-1}$ for $ i-j = m$.
Thus, $H_n = O(\tau^{-2\alpha-1})$.
Moreover, $G_n = O(\min(d,\tau) / n)$.
We obtain the same bound as in case (i), i.e.
    \begin{align*}
        \E \| \hat{\Sigma}_n^\diamond - \ \Sigma_n \|_{F*}^2
        \leq C \left\{ \tau^{-2\alpha-1} + \frac{\min(\tau,d)}{n} \right\}.
    \end{align*}
    Optimal choice of $\tau$ yields the same rates as in case (i).
 \end{proof}
 
\begin{proof}[Proof of Lemma \ref{lem:MSE}] One only needs to verify Condition 2 of \cite{Sancetta2008}. Essentially, this is specific Doukhan-Louhichi type weak dependence condition. It is known that Bernoulli shifts satisfy such weak dependence conditions, but the applying the results in the literature would require to impose stronger assumptions on $G_t$. Therefore, we give a direct proof working under (P.1). Let $ \bar{\epsilon}_{i,j} = (\epsilon_i, \ldots, \epsilon_{j+1}, \epsilon_j', \epsilon_{j-1}', \ldots) $ where $ \{ \epsilon_t' \} $ is an independent copy of $\{ \epsilon_t \} $. Fix $u, v \in \{1, \ldots, 4\} $, $ (i_1,\ldots, i_v), (s_1,\ldots, s_v) \in \N^v $ and $ (j_1,\ldots, j_u), (t_1, \ldots, t_u) \in \N^u$ such that $ s_1 \le \cdots \le s_u + r \le t_1 \le \cdots \le t_v $ for some $ r \in \N$. Put $ U_{t_v}(\beps_{t_v} ) = X_{t_1,j_1} \cdots X_{t_v,j_v}$ and $ V_{s_u}(\beps_{s_u} ) = Y_{s_1,i_1} \cdots Y_{s_u,i_u} $. Then
\begin{align*}
  \Cov( U_{t_v}(\beps_{t_v} ), V_{s_u}(\beps_{s_u} ) ) &
   = \Cov( U_{t_v}( \beps_{t_v-s_u} ), V_{s_u}(\beps_0) ) \\
  &= \Cov( U_{t_v}( \bar\beps_{t_v-s_u,0} ), V_{s_u}(\beps_0) ) + \Cov( U_{t_v}( \beps_{s_u,0} ) - U_{t_v}( \bar\beps_{t_v-s_u,0} ), V_{s_u}(\beps_0) )  \\
  &= \Cov( U_{t_v}( \beps_{s_u,0} ) - U_{t_v}( \bar\beps_{t_v-s_u,0} ), V_{s_u}(\beps_0) )
\end{align*}
By telescoping over the random variables $ U_{t_v}(l) = U_{t_v}( \epsilon_{t_v-s_u}, \ldots, \epsilon_{t_v-s_u-l+1}, \bar\beps_{t_v-s_u-l} ) $, so that $ \sum_{l>t_v-s_u} U_{t_v}(l) - U_{t_v}(l-1) = U_{t_v}( \beps_{s_u,0} ) - U_{t_v}( \bar\beps_{t_v-s_u,0} ) $, and noting that $ \E | (U_{t_v}(l) - U_{t_v}(l-1) ) V_{s_u}(\beps_0) | \le \Theta' l^{-\beta} $ for some constant $ \Theta' $ by (P.1), we obtain
\[
\Cov( U_{t_v}(\beps_{t_v} ), V_{s_u}(\beps_{s_u} ) ) = O( r^{-\beta+1} ),
\]
which establishes the Doukhan-Louhichi type weak dependence condition required in Condition 2 of \cite{Sancetta2008}, since $ \beta > 2$.
\end{proof}

\begin{proof}[Proof of Theorem \ref{ConsistencyApproxErrors}] The first assertion follows from rearranging terms,
	\[
	   \hat{\Sigma}_n^\diamond - \hat{\Sigma}_n^\dagger  - (\Sigma_n^\diamond - \Sigma_n )  = (\hat{\Sigma}_n^\diamond - \Sigma_n^\diamond ) - (\hat{\Sigma}_n^\dagger - \Sigma_n)
	\]
	and the triangle inequality, so that
	\begin{align*}
	\E \| \hat{\Sigma}_n^\diamond - \hat{\Sigma}_n^\dagger  - (\Sigma_n^\diamond - \Sigma_n )  \|_F &\le \E \| \hat{\Sigma}_n^\diamond - \Sigma_n^\diamond \|_F + \E \| \hat{\Sigma}_n^\dagger - \Sigma_n \|_F \\
	& \le \left(  \E \| \hat{\Sigma}_n^\diamond - \Sigma_n^\diamond \|_F^2 \right)^\frac{1}{2} + \left( \E \| \hat{\Sigma}_n^\diamond - \Sigma_n^\diamond \|_F^2  \right)^\frac{1}{2}. 
	\end{align*}
	Hence, the result follows from Theorem~\ref{ConsistencyTapperToeplitz}.
	
	For the second assertion, use the fact that if $ A_n - B_n \le C_n $ on $ A_n \ge B_n $ and $ B_n-A_n \le D_n $ on $ A_n < B_n $, for random variables $A_n, B_n $ and bounds $ C_n, D_n $, then
	\begin{align*}
	  \E | A_n - B_n | &\le \E (A_n-B_n)1_{A_n\ge B_n} + \E (-A_n+B_n)1_{A_n<B_n} \\
	  & \le \E C_n + \E D_n 
	\end{align*}
    Apply this bound with $ A_n = \| \hat{\Sigma}_n^\diamond - \hat{\Sigma}_n^\dagger \|_{F*} $, $ B_n = \| \Sigma_n^\diamond - \Sigma_n^\dagger \|_{F*} $ and the bounds 
\[
     \| \hat{\Sigma}_n^\diamond - \hat{\Sigma}_n^\dagger \|_{F*} - \| \Sigma_n^\diamond - \Sigma_n\|_{F*}  \\
      \le \| \hat{\Sigma}_n^\diamond - \hat{\Sigma}_n^\dagger - ( \Sigma_n^\diamond - \Sigma_n  ) \|_{F*} = C_n
\]
as well as
\[
 \| \Sigma_n^\diamond - \Sigma_n \|_{F*} -\| \hat{\Sigma}_n^\diamond - \hat{\Sigma}_n^\dagger \|_{F*}  \\
\le \| \Sigma_n^\diamond - \Sigma_n  - ( \hat{\Sigma}_n^\diamond - \hat{\Sigma}_n^\dagger ) \|_{F*} = D_n,
\]
which follow from the triangle inequality. Therefore,
\[
  \E \left| \| \hat{\Sigma}_n^\diamond - \hat{\Sigma}_n^\dagger \|_{F*} - \| \Sigma_n^\diamond - \Sigma_n\|_{F*} \right| = o(1),
\]
which implies $ \left| \| \hat{\Sigma}_n^\diamond - \hat{\Sigma}_n^\dagger \|_{F*} - \| \Sigma_n^\diamond - \Sigma_n\|_{F*} \right| \to 0 $, in probability, which in turn gives $\left| \hat{E}_n^\diamond - E_n^\diamond \right| \to 0 $, in probability. 

To show (ii) observe that we have, similarly as above, the bounds
\[
 \pm ( \| \hat{\Delta}_n^\dagger \|_{F*} - \| \Delta_n^\dagger \|_{F*} ) \le \| \hat{\Delta}_n^\dagger - \Delta_n^\dagger \|_{F*}
\]
yielding
\[
  \E \left| \| \hat{\Delta}_n^\dagger \|_{F*} - \| \Delta_n^\dagger \|_{F*} \right|
  \le 2 \E \| \hat{\Delta}_n^\dagger - \Delta_n^\dagger \|_{F*}
  \le 2 \sqrt{ \E \| \hat{\Delta}_n^\dagger - \Delta_n^\dagger \|_{F*}^2 }.
\]
Thus, it suffices to study $ \E \| \hat{\Delta}_n^\dagger - \Delta_n^\dagger \|_{F*}^2 $. We have
\begin{align*}
  \E \| \hat{\Delta}_n^\dagger - \Delta_n^\dagger \|_{F*}^2
  & = d^{-1} \sum_{\tau^\dagger/2 \le m \le \tau^\dagger} \sum_{|i-j|=m} (1-\omega(|i-j|))^2 \E( \hat{\Sigma}_{n,ij} - \Sigma_{n,ij} )^2 \\
  & \qquad + d^{-1} \sum_{\tau^\dagger<m\le \sigma^\dagger} \sum_{|i-j|=m} \E( \hat{\Sigma}_{n,ij} - \Sigma_{n,ij} )^2 \\
  & \qquad +  d^{-1} \sum_{m>\sigma^\dagger} \sum_{|i-j|=m} \Sigma_{n,ij}^2 \\
  & = O( \sigma^\dagger n^{-1} ) + O( (\sigma^\dagger)^{-2\alpha+1} ) \\
  & =O\left( n^{-\frac{2\alpha-1+c-s}{2\alpha+c}} \right) + O\left( n^{-\frac{(s+1)(2\alpha-1)}{2\alpha+c}} \right).
\end{align*}
Both terms are $o(1) $ and the first term dominates if and only if $s>c/2\alpha$. 

(iii) The claim follows directly from the fact that
\[
 \langle \hat{\Delta}_n^\dagger, \hat{\Sigma}_n^\dagger - \hat{\Sigma}_n^\diamond \rangle_* 
 - \langle \Delta_n^\dagger, \Sigma_n^\dagger - \Sigma_n^\diamond \rangle_* 
 = \langle \hat{\Delta}_n^\dagger - \Delta_n^\dagger, \hat{\Sigma}_n^\dagger - \hat{\Sigma}_n^\diamond \rangle_*
 -  \langle \Delta_n^\dagger, \Sigma_n^\dagger - \Sigma_n^\diamond -\hat{\Sigma}_n^\dagger - \hat{\Sigma}_n^\diamond  \rangle_* 
\]
so that
\begin{align*}
\left| \langle \hat{\Delta}_n^\dagger, \hat{\Sigma}_n^\dagger - \hat{\Sigma}_n^\diamond \rangle_* 
 - \langle \Delta_n^\dagger, \Sigma_n^\dagger - \Sigma_n^\diamond \rangle_*   \right| 
 & \quad \le \| \hat{\Delta}_n^\dagger - \Delta_n^\dagger \|_{F*} \| \hat{\Sigma}_n^\dagger - \hat{\Sigma}_n^\diamond  \|_{F*} \\
 & \quad + \| \Sigma_n - \Sigma_n^\dagger \|_{F*} \| \Sigma_n^\dagger - \Sigma_n^\diamond -\hat{\Sigma}_n^\dagger - \hat{\Sigma}_n^\diamond \|_{F*} \\
 & = O( \| \hat{\Delta}_n^\dagger - \Delta_n^\dagger \|_{F*} \| ( \| \Sigma_n - \Sigma_n^\diamond \|_{F*} + o_P(1) ) + o(1) \\
 & = o_P(1),
\end{align*}
since
\[
  \| \hat{\Sigma}_n^\dagger - \hat{\Sigma}_n^\diamond \|_{F*} 
  = \| \Sigma_n - \Sigma_n^\diamond \|_{F*} + o_P(1)
\]
by (i),
\[
  \| \hat{\Delta}_n^\dagger - \Delta_n^\dagger \|_{F*} \to 0,  \quad \text{in probability},
\]
by (ii), 
\[
  \| \Sigma_n - \Sigma_n^\dagger \|_{F*}^2 \le d^{-1} \sum_{m>\tau/2} \sum_{|i-j|=m} \Sigma_{n,ij}^2 = O\left( (\tau/2)^{-2\alpha+1} \right)
  = O\left( n^{-\frac{2\alpha-1}{2\alpha+c}} \right)
\]
and $ \| \Sigma_n - \Sigma_n^\diamond \|_{F*} = O(1) $ by assumption. This completes the proof. 
\end{proof}

\begin{proof}[Proof of Theorem \ref{thm:bootstrap}]
    Note that the statistics $T_n$ and the critical values $a_n(\alpha,w)$ are linear in $\|v\|_1$, such that we may suppose without loss of generality that $\|v\|_1=1$.
    If we replace the $\chi_s$ by the centered random vectors $\bar{\chi}_s = [\nu^j \cdot (X_sX_s^T - \Sigma_0)]_{j=1}^3$, then Theorem \ref{thm:bootstrap} is a special case of \cite[Prop.\ 4.3]{mies2022}.
    We briefly show that the estimation error of $\hat{\Sigma}_n$ versus the true covariance matrix is negligible.
    To this end, denote the sequential estimators of the asymptotic covariance, for $k=b,\ldots, n$, by 
    \begin{align*}
        \hat{Q}(k) = \sum_{t=b}^k \frac{1}{b} \left( \sum_{s=t-b+1}^t \chi_s \right)^{\otimes 2}, \qquad
        \overline{Q}(k) = \sum_{t=b}^k \frac{1}{b} \left( \sum_{s=t-b+1}^t \bar{\chi}_s \right)^{\otimes 2} .
    \end{align*}
    It suffices to show that
    \begin{align}
        \E \max_{k=1,\ldots, n} \|\hat{Q}(k) - \overline{Q}(k)\|_\tr = \mathcal{O}(\sqrt{nb}), \label{eqn:boostrap-estimation}
    \end{align}
    where $\|A\|_\tr = \tr((A A^T)^\frac{1}{2})$ denotes the trace norm of a symmetric matrix.
    Since the matrices $\hat{Q}(k)$ and $\overline{Q}(k)$ are always of fixed dimension $3\times 3$, the choice of norm is in fact irrelevant as all norms are equivalent, and we may use an arbitrary matrix norm $\|\cdot\|$.
    If \eqref{eqn:boostrap-estimation} holds, then both $\hat{Q}(k)$ and $\overline{Q}(k)$ satisfy Theorem 5.1 in \cite{mies2022}, such that the bootstrap scheme based on $\hat{Q}(k)$ is valid.
    
    To establish \eqref{eqn:boostrap-estimation}, denote $\Delta_n = \left[\nu^j \cdot (\Sigma_0 - \hat{\Sigma}_{n,n})\right]_{j=1}^3 = \chi_s - \bar{\chi}_s \in \R^3$, for all $s$. 
    Via quadratic expansion, we find that, for some universal $C$ which may change from line to line,
    \begin{align*}
        &\E \max_{k}\left\| \hat{Q}(k) - \overline{Q}(k) \right\| \\
        &= \E \max_k \left\| \sum_{t=b}^k \frac{1}{b} \left[ \left(\sum_{s=t-b+1}^t (\bar{\chi}_s + \Delta_n) \right)^{\otimes 2} -  \left(\sum_{s=t-b+1}^t \bar{\chi}_s \right)^{\otimes 2} \right] \right\| \\
        &= \E \max_k \left\| \sum_{t=b}^k \frac{1}{b} \left[ \left(\sum_{s=t-b+1}^t (\bar{\chi}_s (b\Delta_n)^T + (b\Delta_n) \bar{\chi}_s^T + (b\Delta_n)^{\otimes 2})\right)  \right] \right\| \\
        &\leq C \, \E \sum_{t=b}^n \left[ 2\left\|\sum_{s=t-b+1}^t \bar{\chi}_s  \Delta_n^T \right\| + b\left\|\Delta_n\right\|^2 \right] \\
        &\leq C\sum_{t=b}^n \sqrt{\E \left\| \sum_{s=t-b+1}^t \bar{\chi}_s \right\|^2} \sqrt{\E \|\Delta_n\|^2} + nb \E \|\Delta_n\|^2.
    \end{align*}
    By virtue of Proposition \ref{prop:matrix} and Lemma \ref{lem:jensen}, the centered time series $\bar{\chi}_s$ satisfies condition (G.1) of \cite{mies2022}, such that the Rosenthal-type inequality, Theorem 3.2 therein, is applicable. 
    This yields,
    \begin{align*}
        \sqrt{\E \left\| \sum_{s=t-b+1}^t \bar{\chi}_s \right\|^2}
        &\leq C \sqrt{b} \sum_{j=1}^\infty \tilde{\Theta} j^{-\beta} = \mathcal{O}(\sqrt{b}), 
    \end{align*}
    because $\tilde{\Theta}$ is fixed in the asymptotic regime of Theorem \ref{thm:bootstrap}.
    Analogously, 
    \begin{align*}
        \sqrt{\E \|\Delta_n\|^2} 
        &= \sqrt{ \E \left\| \tfrac{1}{n} \sum_{t=1}^n \bar{\chi}_t \right\|^2 } = \mathcal{O}(1/\sqrt{n}).
    \end{align*}
    Hence, $\E \max_{k}\left\| \hat{Q}(k) - \overline{Q}(k) \right\| = \mathcal{O}(\sqrt{nb} + b) = \mathcal{O}(\sqrt{nb})$. 
    The last inequality holds because $b\leq n$, and establishes \eqref{eqn:boostrap-estimation}.
\end{proof}

\section*{Acknowledgements}
The authors acknowledge support from Deutsche Forschungsgemeinschaft (DFG, grant STE 1034/11-2).

\bibliography{strong-approximation3}
\bibliographystyle{apalike}

\end{document}